\documentclass[10pt,conference,hidelinks]{IEEEtran}\IEEEoverridecommandlockouts
\usepackage{epsfig,graphicx,subfigure,psfrag,amsmath,cases}
\usepackage{latexsym,amssymb,amsmath,epsfig,subfigure,algorithm,mathtools}
\usepackage{algorithmic}
\usepackage{color}
\usepackage{url}
\usepackage{scrtime}
\usepackage{cite}
\usepackage{subfigure}
\usepackage{bbding}
\usepackage{multicol}

\author{Zhiqiang Wei, Derrick Wing Kwan Ng, and Jinhong Yuan\\
School of Electrical Engineering and Telecommunications, the University of New South Wales\\
Email: zhiqiang.wei@student.unsw.edu.au; w.k.ng@unsw.edu.au; j.yuan@unsw.edu.au\vspace{-5mm}}

\title{\vspace{-3mm}Beamwidth Control for NOMA in Hybrid mmWave Communication Systems\vspace{-2mm}}

\newtheorem{Thm}{Theorem}
\newtheorem{proof}{proof}

\newtheorem{T-Prob}{Transformed Problem}

\newtheorem{Remark}{Remark}

\newcommand{\abs}[1]{\lvert#1\rvert}

\begin{document}
\maketitle
\begin{abstract}
In this paper, we propose a beamwidth control-based non-orthogonal multiple access (NOMA) scheme for hybrid millimeter wave (mmWave) communication systems.
In particular, the proposed scheme allows multiple users in one NOMA group to share the same radio frequency chain and analog beam for superposition transmission.
To overcome the physical limit of the narrow analog beam, a beamwidth control approach is proposed to widen the analog beamwidth to facilitate the formation of NOMA groups.
Then, we characterize the main lobe power loss associated with the proposed beamwidth control and derive the asymptotically optimal analog beamformer to maximize the system sum-rate in the large number of antennas regime.
The system sum-rate gain of the proposed beamwidth control-based NOMA scheme compared to a baseline scheme adopting time division multiple access (TDMA) is analyzed.
Simulation results verify the accuracy of our performance analysis and unveil the importance of beamwidth control for practical mmWave NOMA systems.
\end{abstract}

\section{Introduction}
Millimeter wave (mmWave) communications have recently triggered and attracted tremendous research interests due to their potential in meeting the stringent requirements of ultra-high data rate for the future fifth-generation (5G) wireless networks\cite{Rappaport2013,wong2017key,AkdenizChannelMmWave,Zhang2018low}.
As a result, various mmWave system structures have been proposed in the literature to embrace the promising performance gain\cite{Gao2016Energy,lin2016energy,zhao2017multiuser}.
For example, a fully-access hybrid structure was designed such that each radio frequency (RF) chain is connected to all antennas through an individual group of phase shifters (PSs) \cite{zhao2017multiuser}.
This structure not only reduces substantially the power consumption of RF chains, but also retains the highly directional beamforming gain to compensate the severe propagation loss in mmWave channels.
Despite the rapid development in the transceiver design for mmWave communications\cite{Gao2016Energy,lin2016energy,zhao2017multiuser}, enabling multiple users to share the mmWave spectrum concurrently remains a challenging problem.

Conventionally, mmWave systems adopting orthogonal multiple access (OMA) scheme allocate only a single-user to one RF chain in a time slot\cite{Gao2016Energy,lin2016energy,zhao2017multiuser}.
Recently, non-orthogonal multiple access (NOMA)\cite{QiuLOMA,QiuLOMA2,Sun2016Fullduplex} has been applied to mmWave systems\cite{DingFRAB,Ding2017RandomBeamforming,Cui2017Optimal, WangBeamSpace2017}, which allows more than one users to share a RF chain simultaneously.
It was shown that NOMA can offer a higher spectral efficiency compared to conventional OMA schemes in various types of mmWave communication systems.
For instance, the authors in \cite{DingFRAB} proposed a NOMA transmission scheme in massive MIMO mmWave communication systems via exploiting the features of finite resolution of PSs.
Then, the outage performance analysis for mmWave NOMA systems with random beamforming was studied in \cite{Ding2017RandomBeamforming}.
Subsequently, the authors in \cite{Cui2017Optimal} proposed a user scheduling and power allocation design for the random beamforming mmWave NOMA  scheme.
In \cite{WangBeamSpace2017}, a beamspace mmWave MIMO-NOMA scheme based on a lens antenna array was studied, where a power allocation algorithm was designed to maximize the system sum-rate.
However, the analog beamwidth in mmWave communications is typically very narrow due to its high carrier frequency and the employment of massive number of antennas at transceivers.
Thus, only the users within the main lobe of a beam are suitable for clustering as a NOMA group due to the associated strong channel gains.
In contrast, users outside the main lobe suffer a dramatically signal power attenuation
and cannot maintain a sustainable communication link.
Hence, apply the results from  \cite{DingFRAB,Ding2017RandomBeamforming,Cui2017Optimal, WangBeamSpace2017} to mmWave systems can only create a small number of NOMA groups which limits the potential gain brought by NOMA.

To overcome this limitation, beam splitting schemes were proposed in \cite{Zhu2017Joint,wei2018multibeam,wei2018multiICC}, which offer a considerable spectral efficiency gain compared to OMA and single-beam NOMA schemes.
Nevertheless, the performance of the proposed beam splitting schemes highly depend on the optimal antenna allocation for NOMA users which incurs a prohibitively high computational complexity.
In practice, the beamwidth of an analog beam can be controlled via an appropriate analog beamformer design.
In fact, widening the beamwidth can increase the probability of two users to be covered by the same analog beam which facilitates the exploitation of NOMA gain.
Due to the beamwidth control, there is an inevitable power loss at the main beam direction.
Therefore, there is a non-trivial tradeoff between the beamwidth and the system performance.
Yet, such an investigation has not been reported in the literature.
Therefore, a study of beamwidth control-based hybrid mmWave NOMA scheme is necessary to overcome the limitation of narrow analog beams and is potential to improve the system spectral efficiency.

In this paper, we propose a mmWave NOMA scheme using beamwidth control and analyze the performance of the proposed scheme to provide some system design insights.
To this end, we focus on a simple method which controls the beamwidth of an analog beam by putting some antennas into idle mode.
After characterizing the main lobe power loss associated to beamwidth control, we design the asymptotically optimal analog beamformer to maximize the system sum-rate in the large number of antennas regime.
Our performance analysis unveils some interesting insights about the performance gain of NOMA over OMA in mmWave frequency band.

Notations used in this paper are as follows. Boldface capital and lower case letters are reserved for matrices and vectors, respectively. $\mathbb{C}^{M\times N}$ denotes the set of all $M\times N$ matrices with complex entries; ${\left( \cdot \right)^{\mathrm{T}}}$ denotes the transpose of a vector or a matrix and ${\left( \cdot \right)^{\mathrm{H}}}$ denotes the Hermitian transpose of a vector or a matrix;
$\left\lfloor \cdot \right\rfloor$ denote the maximum integer smaller than the input value; $\abs{\cdot}$ denotes the absolute value of a complex scalar
and ${\rm{E}}\left\{ \cdot \right\}$ denotes the expectation of a random variable.
The circularly symmetric complex Gaussian distribution with mean $\mu$ and variance $\sigma^2$ is denoted by ${\cal CN}(\mu,\sigma^2)$.

\section{System Model}
Consider downlink mmWave communications in a single-cell system with one base station (BS) and two downlink users, as shown in Fig. \ref{HybridNOMAStructure}.
The BS is located at the cell center with a cell radius of $D$ meters.
A commonly adopted uniform linear array (ULA)\cite{zhao2017multiuser} is employed at both the BS and user terminals, where the BS is equipped with $N_{\rm{BS}}$ antennas and each user is equipped with $N_{\rm{UE}}$ antennas.
A single-RF chain\footnote{The extension for the proposed scheme to a multi-RF chain BS and more than two users will be considered in our journal version.} BS with a hybrid fully-access architecture is adopted as a first step to unveil the system design insights of the proposed beamwidth control scheme.
In particular, each RF chain equipped at the BS can access all the $N_{\rm{BS}}$ antennas through $N_{\rm{BS}}$ PSs.
On the other hand, for the sake of a simple receiver design, each user has a single RF chain and connects to all its $N_{\rm{UE}}$ antennas via $N_{\rm{UE}}$ PSs.

\begin{figure}[t]
\centering\vspace{-4mm}
\includegraphics[width=0.40\textwidth]{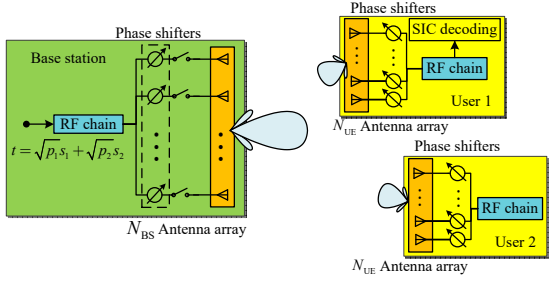}\vspace{-3mm}
\caption{The system model for the proposed beamwidth control-based NOMA mmWave scheme with two users. The proposed scheme exploits switches to activate or idle an antenna and its PS based on the designed beamwidth.}\vspace{-7mm}
\label{HybridNOMAStructure}%
\end{figure}

As shown in Fig. \ref{HybridNOMAStructure}, the key idea of the proposed NOMA scheme is to form a NOMA group through beamwidth control such that the signals of user 1 and user 2 can be superimposed for non-orthogonal transmission through a single-RF chain.
In particular, the condition of forming a NOMA group is that the difference between the angles of departure (AODs) of user 1 and user 2 is within the half-power beamwidth \cite{van2002optimum} of an analog beam from the BS.
Otherwise, there is at least 3 dB of power loss for the user outside the half-power beamwidth.
Hence, efficient communication for this user is hard to guarantee, which leads to a high degree of unfairness.
Therefore, one can imagine that the proposed beamwidth control has a high potential to improve the possibility of forming a NOMA group via widening the analog beam at the BS.
On the other hand, if a NOMA group cannot be formed, a conventional time division multiple access (TDMA) scheme can be employed to provide communications to the two users via allocating the RF chain to user 1 and user 2 within different time slots.
The superimposed signal transmitting through the RF chain at the BS for the formed NOMA group is given by
\vspace{-2mm}
\begin{equation}\label{PureAnalog2}
t = \sqrt{p_1} s_1 + \sqrt{p_2} s_2,\vspace{-2mm}
\end{equation}
where $s_1, s_2 \in \mathbb{C}$ denote the normalized modulated symbols for user 1 and user 2, respectively, with ${\rm{E}}\left\{ {{{\left| {{s_k}} \right|}^2}} \right\} = 1$, $\forall k \in \{1,2\}$.
Variables $p_1$ and $p_2$ denote the power allocation for user 1 and user 2, respectively, satisfying $ p_1 + p_2 \le {p_{{\rm{BS}}}}$ with ${p_{{\rm{BS}}}}$ denoting the maximum transmit power budget for the BS.
In this paper, we consider an equal power allocation to facilitate the performance analysis, i.e., $p_1 = p_2 = \frac{p_{{\rm{BS}}}}{2}$.

On the other hand, we apply the widely adopted Saleh-Valenzuela model \cite{WangBeamSpace2017} for our considered mmWave communication system.
In particular, the channel matrix between the BS and user $k$ ${{\bf{H}}_k} \in \mathbb{C}^{{ {N_{{\rm{UE}}}} \times {N_{{\rm{BS}}}}}}$ is given by
\vspace{-2mm}
\begin{equation}\label{ChannelModel1}
{{\bf{H}}_k} = {\alpha _{k,0}}{{\bf{H}}_{k,0}} + \sum\nolimits_{l = 1}^L {{\alpha _{k,l}}{{\bf{H}}_{k,l}}},\vspace{-2mm}
\end{equation}
where ${\alpha _{k,0}}$ denotes the LOS complex path gain and ${{\bf{H}}_{k,0}}  \in \mathbb{C}^{ N_{\mathrm{UE}} \times N_{\mathrm{BS}} }$ is the LOS channel matrix between the BS and user $k$.
In \eqref{ChannelModel1}, matrix ${\mathbf{H}}_{k,l} \in \mathbb{C}^{ N_{\mathrm{UE}} \times N_{\mathrm{BS}} }$ denotes the $l$-th NLOS channel matrix between the BS and user $k$, ${\alpha _{k,l}}$ denotes the corresponding $l$-th NLOS complex path gain, $1 \le l \le L$, and $L$ denotes the total number of NLOS paths.
In this paper, we define the strong or weak user according to their LOS path gains.
Without loss of generality, we assume that user 1 is the strong user while user 2 is the weak user, i.e., ${\left| {{\alpha _{1,0}}} \right|^2} \ge {\left| {{\alpha _{2,0}}} \right|^2}$.
The channel matrix between user $k$ and the BS for the $l$-th propagation path, ${\mathbf{H}}_{k,l}$, $\forall l \in \{0,\ldots,L\}$, can be generally given by
\vspace{-3.5mm}
\begin{equation}\label{ChannelModel2}
{\mathbf{H}}_{k,l} = {\mathbf{a}}_{\mathrm{UE}} \left(  \phi _{k,l}, {N_{{\mathrm{UE}}}} \right){\mathbf{a}}_{\mathrm{BS}}^{\mathrm{H}}\left( \theta _{k,l}, {N_{{\mathrm{BS}}}} \right),\vspace{-3mm}
\end{equation}
with ${\mathbf{a}}_{\mathrm{BS}}\hspace{-0.5mm}\left(\hspace{-0.5mm} \theta _{k,l}, \hspace{-0.5mm} {N_{{\mathrm{BS}}}} \hspace{-0.5mm} \right) \hspace{-1mm} = \hspace{-1mm} \left[\hspace{-0.5mm} {1, \hspace{-0.5mm}{e^{ - j\hspace{-0.5mm}\frac{2\pi d}{\lambda}\hspace{-0.5mm}   \cos \hspace{-0.5mm}\theta _{k,l} }},\hspace{-0.5mm} \ldots , }{e^{ - j\left(\hspace{-0.5mm}{N_{{\mathrm{BS}}}} \hspace{-0.35mm}-\hspace{-0.35mm} 1\hspace{-0.5mm}\right)\frac{2\pi d}{\lambda}\hspace{-0.5mm}   \cos \hspace{-0.5mm}\theta _{k,l} }}\hspace{-0.5mm}\right]^{\mathrm{T}}$ $ \in \mathbb{C}^{{ {N_{{\mathrm{BS}}}} \times 1}}$
denoting the array response vector \cite{van2002optimum} for the $l$-th path of user $k$ with AOD ${\theta _{k,l}}$ at the BS and
${\mathbf{a}}_{\mathrm{UE}}\left( \phi _{k,l},{N_{{\mathrm{UE}}}} \right) = \left[ 1, {e^{ - j\frac{2\pi d}{\lambda} \cos \phi _{k,l} }}, \ldots ,{e^{ - j{\left({N_{{\mathrm{UE}}}} - 1\right)}\frac{2\pi d}{\lambda} \cos \phi _{k,l} }} \right] ^ {\mathrm{T}} \in \mathbb{C}^{{ {N_{{\mathrm{UE}}}} \times 1}}$ denoting the array response vector for the $l$-th path with angle of arrival (AOA) ${\phi _{k,l}}$ at user $k$.
Note that $\lambda$ denotes the wavelength of the carrier frequency and $d = \frac{\lambda}{2}$ denotes the space between adjacent antennas.
In addition, we assume that the LOS channel state information (CSI), including the AODs ${\theta _{k,0}}$ and the complex path gains ${{\alpha _{k,0}}}$ of both users, is known at the BS via some effective beam tracking techniques, e.g. \cite{BeamTracking}.
Due to the same reason, we can assume that the AOA, ${\phi _{k,0}}$, is known at user $k$, $\forall k \in \{1,2\}$.
In addition, every user would like to steer its receive analog beamformer to its LOS AOA to maximize the desired received signal power and thus we have ${\bf{v}}_k = \frac{1}{\sqrt{N_{{\mathrm{UE}}}}} {\mathbf{a}}_{\mathrm{UE}} \left(  \phi _{k,l}, {N_{{\mathrm{UE}}}} \right)$, $\forall k$.
Denoting ${\bf{w}} \in \mathbb{C}^{{{N_{{\rm{BS}}}} \times 1}}$ as the designed analog beamformer at the BS according the adopted beamwidth control method, with ${\left\| {{{\bf{w}}}} \right\|^2} = 1$, the effective channel of user $k$ can be denoted by
\vspace{-2mm}
\begin{equation}\label{EffectiveChannel}
\widetilde{{h}}_k = {\bf{v}}_k^{\rm{H}}{{\bf{H}}_k}{\bf{w}} \in \mathbb{C}.\vspace{-2mm}
\end{equation}

At the user side, the received signal at user $k$ after processed by the receive analog beamformer is given by
\vspace{-2mm}
\begin{equation}\label{SystemModel}
{y_k} = {\bf{v}}_k^{\rm{H}}{{\bf{H}}_k}{\bf{w}}t + {\bf{v}}_k^{\rm{H}}{{\bf{z}}_k}
=\widetilde{h}_k{t} + {\bf{v}}_k^{\rm{H}}{{\bf{z}}_k}, \forall k,\vspace{-2mm}
\end{equation}
where $\mathbf{z}_k \in \mathbb{C}^{{N_{\rm{UE}} \times 1}}$ is the additive white Gaussian noise (AWGN) at receiving antenna array of user $k$, i.e., $\mathbf{z}_k \sim {\cal CN}(0,\sigma^2 \mathbf{I}_{N_{\rm{UE}}})$ with $\sigma^2$ denoting the noise power.
Similar to the traditional downlink NOMA schemes\cite{WeiTCOM2017}, SIC decoding is performed at the strong user within a NOMA group, while the weak user directly decodes its desired messages by treating the strong user's signal as noise.

\section{Analog Beamformer Design with Beamwidth Control}
Clearly, the key challenge of the proposed scheme is to design an effective analog beamformer taking into account the beamwidth control.
In this section, we first present the proposed beamwidth control method and characterize the associated power loss in the main lobe.
Then, the asymptotically optimal analog beamformer to maximize the system sum-rate is derived via determining the asymptotically optimal main beam direction and the desired beamwidth sequentially.

\begin{figure}[t]
\centering\vspace{-4mm}
\includegraphics[width=2.3in]{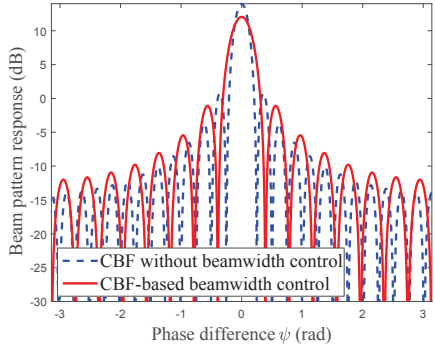}\vspace{-4mm}
\caption{The beam pattern response with/without beamwidth control. The number of antennas at the BS is $N_{\rm{BS}} = 25$ and the main lobe direction is $\theta_0 = 90$ degrees. The desired half-power beamwidth of CBF with beamwidth control is 1.5 times of that of CBF without beamwidth control.}\vspace{-6mm}
\label{BeamPatternWithBeamWidthControl}
\end{figure}

\subsection{Proposed Beamwidth Control}
The most natural and simple way of beamwidth control is to reduce the number of active antennas to widen the analog beamwidth in the hybrid architecture with constant-modulus PSs\footnote{For the case with amplitude-adjustable PSs, more beamwidth control methods can be utilized, which will be detailed in our journal version.}, as shown in Fig. \ref{HybridNOMAStructure}.
Considering an ULA with $N_{\rm{BS}}$ antennas but only $N'_{\rm{BS}}$ consecutive antennas of them are active.
The conventional beamformer (CBF) to generate an analog beam steering towards AOD $\theta_0$ is given by
\vspace{-2mm}
\begin{align}\label{CBFWeight}
&{{\bf{w}}_{{\rm{CBF}}}}\left(\theta_0,N'_{\rm{BS}}\right) \\[-1mm] &=\frac{1}{\sqrt{N'_{\rm{BS}}}}\hspace{-1mm}\left[\hspace{-0.5mm} {1,\hspace{-0.5mm}{e^{ - j\hspace{-0.5mm}\frac{2\pi d}{\lambda}\hspace{-0.5mm}   \cos \hspace{-0.5mm}\theta _{0} }},\hspace{-0.5mm}\ldots,\hspace{-0.5mm}{e^{ - j\left(\hspace{-0.5mm}N'_{\rm{BS}} \hspace{-0.35mm}-\hspace{-0.35mm} 1\hspace{-0.5mm}\right)\frac{2\pi d}{\lambda}\hspace{-0.5mm}   \cos \hspace{-0.5mm}\theta _{0} }}},0,\ldots,0 \right]^{\rm{T}}\hspace{-0.5mm}, \notag
\end{align}
\vspace{-4mm}
\par
\noindent
where ${{{\bf{w}}_{{\rm{CBF}}}}} \in \mathbb{C}^{{N_{\rm{BS}}} \times 1}$, ${\left\| {{{\bf{w}}_{{\rm{CBF}}}}} \right\|^2} = 1$, and its last ${N_{\rm{BS}}} - N'_{\rm{BS}}$ entries are zero means that those antennas are inactive.
The main lobe direction $\theta_0$ and the number of active antennas $N'_{\rm{BS}}$ will be determined latter in this section.
The beam response function characterizes the beamforming gain for the signal departing from an arbitrary direction $\theta$.
In particular, the beam response function of the CBF is obtained as\cite{van2002optimum}:
\vspace{-2mm}
\begin{equation}\label{BeamResponseFunction}
\left|{B_{{\rm{CBF}}}}\left( \psi  \right)\right|^2 = \frac{1}{{N'_{\rm{BS}}}}\frac{{\sin^2 \left( {\frac{{N'_{\rm{BS}}\psi }}{2}} \right)}}{{\sin^2 \left( {\frac{\psi }{2}} \right)}},\vspace{-2mm}
\end{equation}
where $\psi = \frac{2\pi d}{\lambda}\left(\cos\left(\theta\right) - \cos\left(\theta_0\right)\right)$ denotes the phase difference\footnote{Note that, the beam response function in the angle domain $\left|{B_{{\rm{CBF}}}}\left( \theta,\theta_0  \right)\right|$  can be easily obtained via the mapping relationship $\psi = \frac{2\pi d}{\lambda}\left(\cos\left(\theta\right) - \cos\left(\theta_0\right)\right)$.} between the signals departing from directions $\theta_0$ and $\theta$.
In addition, the half-power beamwidth (i.e., -3 dB beamwidth) can be approximated by\cite{van2002optimum}
\vspace{-2mm}
\begin{equation}\label{Beamwidth3dBCBF}
{\psi^{\rm{CBF}}_{\rm{H}}} \approx 0.891\frac{2\pi}{N'_{\rm{BS}}}.\vspace{-2mm}
\end{equation}
According to \eqref{BeamResponseFunction}, the main lobe response for CBF is given by
\vspace{-3mm}
\begin{equation}\label{MLR}
G_{\rm{CBF}} = \left|{B_{{\rm{CBF}}}}\left( 0 \right)\right|^2 = N'_{\rm{BS}}.\vspace{-2mm}
\end{equation}
Combining \eqref{Beamwidth3dBCBF} and \eqref{MLR}, we have
\vspace{-2mm}
\begin{equation}\label{MLRVersus3dBBeamwdith}
G_{\rm{CBF}} \approx 0.891\frac{2\pi}{\psi^{\rm{CBF}}_{\rm{H}}}.\vspace{-2mm}
\end{equation}

It can be observed from \eqref{MLRVersus3dBBeamwdith} that the main lobe response $G_{\rm{CBF}}$ is a monotonically decreasing function with an increasing half-power beamwidth for the CBF.
This is due to the conservation of energy as the beam patterns with different half-power beamwidths have the same amount of total radiation energy.
In fact, it is the expense for widening the analog beamwidth to accommodate a NOMA group.
The tradeoff between the analog beamwidth and the main lobe response can also be observed in Fig. \ref{BeamPatternWithBeamWidthControl}.
Hence, the combination of NOMA with beamwidth control offers a higher flexibility in serving multiple users compared to traditional scheme such as TDMA.
Besides, as observed from \eqref{MLR} and \eqref{MLRVersus3dBBeamwdith}, there is a simple one-to-one mapping between the half-power beamwidth ${\psi^{\rm{CBF}}_{\rm{H}}}$ and the number of active antennas ${N'_{\rm{BS}}}$, and we can determine ${N'_{\rm{BS}}}$ for achieving a desired ${\psi^{\rm{CBF}}_{\rm{H}}}$.

\subsection{Proposed Analog Beamformer}
\subsubsection{Main Beam Direction}
Given the desired beamwidth ${\psi^{\rm{CBF}}_{\rm{H}}}$ or equivalently the number of active antennas ${N'_{\rm{BS}}}$, we need to determine the main beam direction $\theta_0$ of the analog beamformer for the NOMA group.
In particular, rotating the analog beam changes the effective channel gains of the two NOMA users, and thus affects the SIC decoding order as well as the sum-rate.
In this section, for an arbitrary $N'_{\rm{BS}}$, we derive the asymptotically optimal main beam direction to maximize the sum-rate of a NOMA group.

As mentioned before, only LOS CSI is assumed to be known for the analog beamformer design at the BS.
Therefore, to facilitate the analog beamformer design and the performance analysis, we assume a pure LOS mmWave channel in this section, i.e., ${{\bf{H}}_k} = {\alpha _{k,0}}{{\bf{H}}_{k,0}}$.
Note that it is a reasonable assumption since the LOS path gain is generally much stronger than the NLOS path gains\footnote{The system sum-rate taking into account the NLOS paths will be shown in the simulation section.}\cite{Rappaport2013}.

In the proposed scheme, changing the main beam direction $\theta_0$ of the generated analog beam may reverse the effective channel gain order of the two users, which in turn reverses the SIC decoding order and changes the sum-rate.
For the first case, if a selected main beam direction $\widetilde{\theta}_0$ does not change the channel gain order of the two users, we have $\left| {\widetilde h_1} \left(\widetilde{\theta}_0\right) \right| \ge \left| {\widetilde h_2} \left(\widetilde{\theta}_0\right) \right|$, where
${\widetilde{h}_k}\left(\widetilde{\theta}_0\right) = {\bf{v}}_k^{\rm{H}}{{\bf{H}}_k}{{\bf{w}}_{{\rm{CBF}}}}\left(\widetilde{\theta}_0,N'_{\rm{BS}}\right)$ denotes the effective channel of user $k$, $\forall k \in \{1,2\}$.
According to the downlink NOMA protocol\cite{WeiTCOM2017}, the achievable sum-rate of both users is given by
\vspace{-3mm}
\begin{align}\label{SumRate1}
{R}_{{\rm{sum}}}\left(\widetilde{\theta}_0\right) &= {\log _2}\left( {1 + \rho{{{p_1}{{\left| {\widetilde h_1} \left(\widetilde{\theta}_0\right) \right|}^2}}}} \right) \notag\\[-1mm]
&+ {\log _2}\left( {1 + \frac{{\rho{p_2}{{\left| {\widetilde h_2} \left(\widetilde{\theta}_0\right) \right|}^2}}}{{\rho{p_1}{{\left| {\widetilde h_2}\left(\widetilde{\theta}_0\right) \right|}^2} + 1}}} \right),
\end{align}
\vspace{-4mm}\par\noindent
where $\rho = \frac{1}{\sigma^2}$ and the two terms denote the achievable data rates of user 1 and user 2, respectively\footnote{In the considered single-RF chain system, the SIC decoding at the strong user, i.e., user 1, can be guaranteed when $\left| {\widetilde h_1}\left(\widetilde{\theta}_0\right) \right| \ge \left| {\widetilde h_2}\left(\widetilde{\theta}_0\right) \right|$\cite{Tse2005}.}.
On the other hand, in the second case, if a selected main beam direction, $\overline{\theta}_0$, does reverse the channel gain order of the two users, i.e., $\left| {\widetilde h_2}\left(\overline{\theta}_0\right) \right| \ge \left| {\widetilde h_1} \left(\overline{\theta}_0\right)\right|$, similarly, the achievable sum-rate can be written as:
\vspace{-3mm}
\begin{align}\label{SumRate2}
\overline{{R}}_{{\rm{sum}}} \left(\overline{\theta}_0\right) &= {\log _2}\left( {1 + \rho{{{p_2}{{\left| {\widetilde h_2}\left(\overline{\theta}_0\right) \right|}^2}}}} \right) \notag\\[-1mm]
&+ {\log _2}\left( {1 + \frac{{\rho{p_1}{{\left| { \widetilde h_1}\left(\overline{\theta}_0\right) \right|}^2}}}{{\rho{p_2}{{\left| { \widetilde h_1}\left(\overline{\theta}_0\right) \right|}^2} + 1}}} \right).
\end{align}
\vspace{-4mm}\par\noindent
Now, the asymptotically optimal main beam direction is stated in the following theorem.

\begin{Thm}\label{BeamDirection}
For any desired beamwidth: 1) the main beam direction keeping the original channel gain order outperforms the one that reversing the channel gain order in terms of system sum-rate, i.e., $\mathop {\max }\limits_{{\theta}_0}  \mathop {\lim }\limits_{N_{\rm{BS}} \to \infty } {R}_{{\rm{sum}}}\left( {\theta}_0  \right) \ge \mathop {\max }\limits_{{\theta}_0}  \mathop {\lim }\limits_{N_{\rm{BS}} \to \infty } \overline{{R}}_{{\rm{sum}}}\left( {\theta}_0  \right)$;  2) steering the beam to the strong user is asymptotically optimal in terms of the system sum-rate in the large number of antennas regime, i.e., $\theta_{0,1} = \arg\mathop {\max }\limits_{{\theta}_0}  \mathop {\lim }\limits_{N_{\rm{BS}} \to \infty } R_{{\rm{sum}}}\left( {\theta}_0  \right)$.
\end{Thm}

\begin{proof}
Firstly, it is clear that the optimal main beam direction should be located within the range of angles spanned by the two NOMA users, i.e., $\theta_0 \in \left[\theta_{0,1},\theta_{0,2}\right]$, as steering the beam outside this range degrades both users' effective channel gains.
Since only the two users within the half-power beamwidth can be clustered as a NOMA group, in the large number of antennas regime, we have $\mathop {\lim }\limits_{N_{\rm{BS}} \to \infty } \rho{\left| {\widetilde h_1}\left({\theta}_0\right) \right|^2} \to \infty$ and $\mathop {\lim }\limits_{N_{\rm{BS}} \to \infty } \rho{\left| {\widetilde h_2}\left({\theta}_0\right) \right|^2} \to \infty$, $\forall \theta_0 \in \left[\theta_{0,1},\theta_{0,2}\right]$.
As a result, the asymptotic sum-rates for both cases in \eqref{SumRate1} and \eqref{SumRate2} can be rewritten as
\vspace{-2mm}
\begin{align}
\mathop {\lim }\limits_{N_{\rm{BS}} \to \infty } R_{{\rm{sum}}} \left(\widetilde{\theta}_0\right) &= {\log _2}\left( \rho{{{{p_1}{{\left| {\widetilde h_1} \hspace{-0.5mm} \left(\widetilde{\theta}_0\right) \right|}^2}}}} \right)
\hspace{-0.5mm}+\hspace{-0.5mm} {\log _2}\left( {1\hspace{-0.5mm} + \hspace{-0.5mm}\frac{p_2}{p_1}} \right) \notag\\[-0.75mm]
& = {\log _2}\left( \rho{{{{p_{{\rm{BS}}}}{{\left| {\widetilde h_1} \hspace{-0.5mm} \left(\widetilde{\theta}_0\right) \right|}^2}}}} \right) \;\text{and}\label{HighSNRRateNOMApROOF1}\\[-0.75mm]
\mathop {\lim }\limits_{N_{\rm{BS}} \to \infty } \overline{R}_{{\rm{sum}}} \left(\overline{\theta}_0\right) &={\log _2}\left( {\rho{{{p_2}{{\left| {\widetilde h_2}\left(\overline{\theta}_0\right) \right|}^2}}}} \right)
\hspace{-0.5mm}+\hspace{-0.5mm} {\log _2}\left( {1 \hspace{-0.5mm}+\hspace{-0.5mm} \frac{{{p_1}}}{{{p_2}}}} \right)\notag\\[-0.75mm]
&={\log _2}\left( {\rho{{{p_{{\rm{BS}}}}{{\left| {\widetilde h_2}\left(\overline{\theta}_0\right) \right|}^2}}}} \right),\label{HighSNRRateNOMApROOF2}
\end{align}
\vspace{-4mm}\par\noindent
respectively.
For both cases, in the large number of antennas regime, we can observe that the sum-rate of the two NOMA users only scales with the strength of the larger effective channel gain after analog beamforming.
Therefore, we have
\vspace{-6mm}
\begin{align}
\theta_{0,1} &= \arg\mathop { \max}\limits_{{\theta}_0}  \mathop {\lim }\limits_{N_{\rm{BS}} \to \infty } R_{{\rm{sum}}}\left( {\theta}_0  \right) \;\text{and}\; \label{OptimalMainBeamDirection}\\[-1mm]
\theta_{0,2} &= \arg\mathop { \max}\limits_{{\theta}_0}  \mathop {\lim }\limits_{N_{\rm{BS}} \to \infty } \overline{R}_{{\rm{sum}}}\left( {\theta}_0  \right),\label{OptimalMainBeamDirection2}
\end{align}
\vspace{-4mm}\par\noindent
as $\theta_{0,1}$ and $\theta_{0,2}$ can maximize the effective channel gain ${\left| {\widetilde h_1} \left({\theta}_0\right) \right|^2}$ and ${\left| {\widetilde h_2} \left({\theta}_0\right) \right|^2}$ in \eqref{HighSNRRateNOMApROOF1} and \eqref{HighSNRRateNOMApROOF2}, respectively.
The maximum effective channel gains of user 1 and user 2 in the considered two cases can be obtained by
\vspace{-2mm}
\begin{align}
{\left| {\widetilde h_1} \left({\theta}_{0,1}\right) \right|^2} &= {\left|{\alpha _{1,0}}\right|}^2 {{N_{{\rm{UE}}}}G_{\rm{CBF}}} \;\text{and}\; \label{EffectiveChannelMax1}\\[-1mm]
{\left| {\widetilde h_2} \left({\theta}_{0,2}\right) \right|^2} &= {\left|{\alpha _{2,0}}\right|}^2 {{N_{{\rm{UE}}}}G_{\rm{CBF}}},\label{EffectiveChannelMax2}
\end{align}
\vspace{-4mm}\par\noindent
respectively, where the receiving beamforming gain ${N_{{\rm{UE}}}}$ is obtained via ${\bf{v}}_k = \frac{1}{\sqrt{N_{{\mathrm{UE}}}}}{\mathbf{a}}_{\mathrm{UE}} \left(  \phi _{k,0}, {N_{{\mathrm{UE}}}} \right)$, $\forall k = \{1,2\}$.
Here, $G_{\rm{CBF}}$ denotes the main lobe response of transmitting analog beamforming, which is determined by the adopted beamwidth.
Moreover, rotating a beam does not change its main lobe response.
Hence, we have the same main lobe response for the analog beams steered to user 1 and user 2, respectively.
Furthermore, due to ${{\left| {{\alpha _{2,0}}} \right|}} < {{\left| {{\alpha _{1,0}}} \right|}}$, $\mathop {\max }\limits_{{\theta}_0}  \mathop {\lim }\limits_{N_{\rm{BS}} \to \infty } {R}_{{\rm{sum}}}\left( {\theta}_0  \right) \ge \mathop {\max }\limits_{{\theta}_0}  \mathop {\lim }\limits_{N_{\rm{BS}} \to \infty } \overline{{R}}_{{\rm{sum}}}\left( {\theta}_0  \right)$ holds, which completes the proof of the first part of Theorem \ref{BeamDirection}.
Then, combining this result with \eqref{OptimalMainBeamDirection}, we can conclude that steering the generated analog beam to the strong user, i.e., user 1, can maximize the system sum-rate in the large number of antennas regime, which completes the proof of the second part of Theorem \ref{BeamDirection}.
\end{proof}

\begin{Remark}
The results obtained in this paper provide generalization and justification to the commonly adopted distance-based user pairing, e.g. \cite{Dingtobepublished}, in NOMA systems.
\end{Remark}
\subsubsection{The desired beamwidth}
After determining the asymptotically optimal main beam direction for an arbitrary analog beam, the optimal desired half-power beamwidth can be obtained based on the phase difference $\psi_{12} = \frac{2\pi d}{\lambda}\left(\cos\left(\theta_{0,1}\right) - \cos\left(\theta_{0,2}\right)\right)$ between user 1 and user 2.
Since the system sum-rate is proportional to the main lobe response in \eqref{HighSNRRateNOMApROOF1} while the main lobe response is inversely proportional to the desired beamwidth in \eqref{Beamwidth3dBCBF}.
Therefore, the desired beamwidth should be as small as possible to maximize the system sum-rate.
When all the ${N_{\rm{BS}}}$ antennas are active, the minimum half-power beamwidth is given by $\psi _{\rm{H}}^{\rm{min}} = 0.891\frac{2\pi}{N_{\rm{BS}}}$, i.e.,
without applying any beamwidth control.
Therefore, when $\left|{\psi _{12}}\right| \le \frac{{\psi^{\rm{min}} _{\rm{H}}}}{2}$, increasing the beamwidth is not necessary.
The analog beamformer is given by CBF with the whole antenna array ${\bf{w}} = {{\bf{w}}_{{\rm{CBF}}}}\left( {{\theta _{1,0}},N_{{\rm{BS}}}} \right)$ and there is no power loss at the main beam direction.
On the other hand, when $\frac{{\psi^{\rm{min}} _{\rm{H}}}}{2} < \left|{\psi _{12}}\right|$, we need to reduce the number of activated antennas such that $\frac{\psi _{\rm{H}}}{2} = 0.891\frac{\pi}{N'_{{\rm{BS}}}} \ge \left|{\psi _{12}}\right|$.
In this case, the optimal number of active antennas should be ${N'_{\rm{BS}}} = \left\lfloor 0.891\frac{\pi}{\left|{\psi _{12}}\right|}\right\rfloor$ to minimize the main lobe power loss and the analog beamformer is given by ${\bf{w}} = {{\bf{w}}_{{\rm{CBF}}}}\left( {{\theta _{1,0}},N'_{\rm{BS}}} \right)$.
In summary, the asymptotically optimal number of active antennas or equivalently the main lobe response can be given by
\vspace{-2mm}
\begin{equation}\label{OptimalMLR}
G_{\rm{CBF}} = \min \left({N_{{\rm{BS}}}}, \left\lfloor 0.891\frac{\pi}{\left|{\psi _{12}}\right|}\right\rfloor\right).\vspace{-1mm}
\end{equation}
It can be observed from \eqref{OptimalMLR} that the larger the phase difference between the two users, the smaller the main lobe response of the analog beam to accommodate the formed NOMA group.
Therefore, one can imagine that we prefer to cluster two users with similar AODs as a NOMA group to reduce the incurred the main lobe power loss.
Note that, the proposed scheme is suitable for various practical communication scenarios where the users experience the similar AODs from the transmitter, such as a cluster of wireless sensors located closely.

\section{Performance analysis}
In this section, we analyze the sum-rate performance of the proposed mmWave NOMA scheme with the adopted beamwidth control.
More importantly, a sufficient condition for the proposed scheme offering a superior spectral efficiency compared to the baseline TDMA scheme is derived.

Substituting \eqref{EffectiveChannelMax1} and \eqref{OptimalMLR} into \eqref{HighSNRRateNOMApROOF1}, we can obtain the asymptotic system sum-rate of the proposed NOMA scheme in the high number of antennas regime as follows:
\vspace{-2mm}
\begin{align}\label{SumRateNOMA}
&\mathop {\lim }\limits_{N_{\rm{BS}} \to \infty } R^{\rm{NOMA}}_{{\rm{sum}}} \notag\\[-1mm]
&=\hspace{-1mm}{\log _2}\hspace{-1mm}\left( {\rho{{{p_{{\rm{BS}}}}{\left|{\alpha _{1,0}}\right|}^2 {N_{{\rm{UE}}}}\min \hspace{-1mm}\left({N_{{\rm{BS}}}}, \left\lfloor 0.891\frac{\pi}{\left|{\psi _{12}}\right|}\right\rfloor\right)\hspace{-1mm}}}} \right).
\end{align}
\vspace{-3mm}\par\noindent

Considering a baseline TDMA scheme with equal time allocation for the two users to share the single-RF chain, the sum-rate can be given by
\vspace{-2mm}
\begin{align}\label{SumRateOMA}
\hspace{-3mm}\mathop {\lim }\limits_{N_{\rm{BS}} \to \infty } R^{\rm{TDMA}}_{{\rm{sum}}} &=\frac{1}{2}{\log _2}\left( 1+ {\rho{{p_{{\rm{BS}}}}{\left|{\alpha _{1,0}}\right|}^2 {{N_{{\rm{UE}}}}{N_{{\rm{BS}}}}}}} \right) \notag\\[-0.5mm]
\hspace{-3mm}&+ \frac{1}{2}{\log _2}\left( 1+ {\rho{{p_{{\rm{BS}}}}{\left|{\alpha _{2,0}}\right|}^2 {{N_{{\rm{UE}}}}{N_{{\rm{BS}}}}}}} \right),
\end{align}
\vspace{-4mm}\par\noindent
where the two terms denote the achievable data rates of user 1 and user 2 for the baseline scheme, respectively.
In the large number of antennas regime, i.e., ${N_{\rm{BS}} \to \infty }$, the asymptotic sum-rate of the baseline TDMA scheme is given by
\vspace{-2mm}
\begin{equation}\label{AsySumRateOMA}
\mathop {\lim }\limits_{N_{\rm{BS}} \to \infty } R^{\rm{TDMA}}_{{\rm{sum}}} ={\log _2}\left( {\rho{{p_{{\rm{BS}}}}{\left|{\alpha _{1,0}}\right|}{\left|{\alpha _{2,0}}\right|} {{N_{{\rm{UE}}}}{N_{{\rm{BS}}}}}}} \right).\vspace{-2mm}
\end{equation}

Comparing \eqref{SumRateNOMA} and \eqref{AsySumRateOMA}, the system sum-rate gain of the proposed NOMA scheme over the baseline TDMA scheme can be given by
\vspace{-2mm}
\begin{align}
&\mathop{\lim }\limits_{N_{\rm{BS}} \to \infty } R^{\rm{NOMA}}_{{\rm{sum}}} - R^{\rm{TDMA}}_{{\rm{sum}}}\notag\\[-1mm]
&= {\log _2}\left(
\frac{ {\alpha _{12}} \min \left({N_{{\rm{BS}}}}, \left\lfloor 0.891\frac{\pi}{\left|{\psi _{12}}\right|}\right\rfloor\right)}{{N_{{\rm{BS}}}}} \right),
\end{align}
\vspace{-3mm}\par\noindent
where ${\alpha _{12}} = \frac{{\left|{\alpha _{1,0}}\right|}}{{\left|{\alpha _{2,0}}\right|}} \ge 1$ denotes the LOS path gain ratio between user 1 and user 2.
Therefore, the sufficient condition for the proposed scheme outperforming the baseline TDMA scheme, i.e., $\mathop{\lim }\limits_{N_{\rm{BS}} \to \infty } R^{\rm{NOMA}}_{{\rm{sum}}} - R^{\rm{TDMA}}_{{\rm{sum}}} \ge 0$, can be given by
\vspace{-1mm}
\begin{equation}\label{SuffCondition}
{ {\alpha _{12}} \min \left({N_{{\rm{BS}}}}, \left\lfloor 0.891\frac{\pi}{\left|{\psi _{12}}\right|}\right\rfloor\right)} \ge {{N_{{\rm{BS}}}}}.\vspace{-1mm}
\end{equation}
It can be observed that when ${\left|{\psi _{12}}\right|} \le {\psi^{\rm{min}} _{\rm{H}}}$, the sufficient condition in \eqref{SuffCondition} is guaranteed to be satisfied.
In other words, the proposed NOMA scheme outperforms the baseline TDMA scheme when both users have similar AODs.
Besides, when ${\left|{\psi _{12}}\right|} > {\psi^{\rm{min}} _{\rm{H}}}$, the proposed NOMA scheme has a sum-rate gain over the baseline TDMA scheme only if ${ {\alpha _{12}}\left\lfloor 0.891\frac{\pi}{\left|{\psi _{12}}\right|}\right\rfloor} \ge {{N_{{\rm{BS}}}}}$.
Therefore, we can define a \emph{sum-rate gain region} in terms of both ${\psi _{12}}$ and ${\alpha _{12}}$, as will be shown in the simulation section, wherein the proposed scheme has a superior spectral efficiency performance than the baseline TDMA scheme.

\section{Simulation Results}
In this section, we use simulations to evaluate the performance of our proposed beamwidth control-based mmWave NOMA scheme.
We consider a hybrid mmWave communication system with a carrier frequency at $28$ GHz.
There are one LOS path and $L = 10$ NLOS paths for the channel model in \eqref{ChannelModel1} and the path loss models for LOS and NLOS paths follow Table I in \cite{AkdenizChannelMmWave}.
Both users are randomly and uniformly distributed in the single-cell with a cell radius of $D = 200$ meters with their maximum AOD difference as $10$ degrees\footnote{Since the two users with similar AODs are preferred to be grouped as a NOMA group, the users with AOD difference larger than $10$ degrees can be excluded via user scheduling.}.
The maximum transmit power of the BS is 30 dBm, i.e., ${p_{{\rm{BS}}}} = 30$ dBm and the noise power at both users is assumed identical with $\sigma^2 = -88$ dBm.
The number of antennas equipped at the BS $N_{\mathrm{BS}}$ ranges from $16$ to $64$ and the number of antennas equipped at each user terminal is $N_{\mathrm{UE}} = 8$.
\subsection{Sum-rate Gain Region and the Probability of Forming a NOMA Group}
\begin{figure}[t]
\centering\vspace{-4mm}
\includegraphics[width=2.8in]{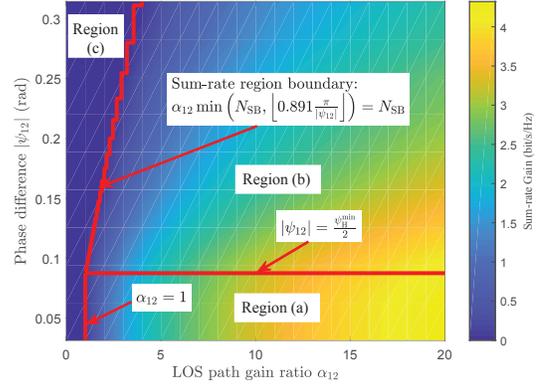}\vspace{-4mm}
\caption{An illustration of the sum-rate gain region of the proposed scheme with respect to the baseline TDMA scheme.}
\label{NOMARegion_CBF}
\end{figure}

\begin{figure}[t]
\centering\vspace{-4mm}
\includegraphics[width=2.7in]{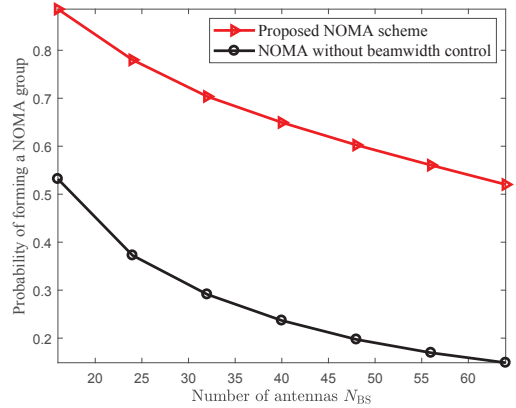}\vspace{-4mm}
\caption{The probability of forming a NOMA group versus the number of antennas equipped at the BS with/without beamwidth control.}\vspace{-6mm}
\label{ProbabilityNOMA}
\end{figure}

Considering two users with different LOS path gain ratio and different phase differences, Fig. \ref{NOMARegion_CBF} illustrates the \emph{sum-rate gain region} (denoted by region (a) and region (b)) of the proposed scheme with respect to the baseline TDMA scheme with $N_{\rm{BS}} = 32$.
The derived analytical results determining the region boundaries are indicated with equations and shown with red lines in Fig. \ref{NOMARegion_CBF}.
The simulation results of the sum-rate gain of the proposed scheme over the baseline TDMA scheme are also illustrated for comparison.
It can be observed that there is no sum-rate improvement of the proposed scheme in region (c) due to the main lobe power loss and hence beamwidth control is not needed in this region.
We can also observe that the sufficient condition derived in \eqref{SuffCondition} can accurately predict whether the proposed NOMA scheme offers a higher sum-rate compared to the baseline TDMA scheme or not.
In addition, it can be observed that the horizonal line $\left|{\psi _{12}}\right| = \frac{{\psi^{\rm{min}} _{\rm{H}}}}{2}$ divides the sum-rate gain region into two parts: region (a) and region (b).
In fact, region (a) is the sum-rate gain region of the conventional NOMA scheme without beamwidth control.
The proposed NOMA scheme substantially extends the sum-rate gain region from (a) to (b) through controlling the beamwidth effectively, which leads to a higher probability of forming a NOMA group.

The probability of forming a NOMA group versus the number of antennas equipped at the BS $N_{\rm{BS}}$ is shown in Fig. \ref{ProbabilityNOMA}.
We can observe that the probability of forming a NOMA group can be increased substantially with employing the proposed beamwidth control method.
In addition, the probability of forming a NOMA group decreases with increasing $N_{\rm{BS}}$.
In fact, the larger array employed at the BS, the smaller the analog beamwidth, and thus the smaller the probability of forming a NOMA group.

\subsection{Average Sum-rate versus the Number of Antennas at the BS}

Fig. \ref{SumRateSchemes} depicts the average system sum-rate of the proposed scheme versus the number of antennas equipped at the BS $N_{\rm{BS}}$.
The performances for the conventional NOMA scheme without beamwidth control and the baseline TDMA scheme are also shown for comparison.
We can observe that the simulation result matches tightly with our performance analysis even though the number of antennas is finite.
In addition, both NOMA schemes outperform the baseline TDMA scheme since NOMA allows the two users to share a RF chain simultaneously in hybrid mmWave communication systems.
It can be observed that the performance gain of NOMA without beamwidth control over the TDMA scheme decreases with an increasing $N_{\rm{BS}}$.
In fact, the probability of forming a NOMA group without beamwidth control vanishes dramatically with increasing $N_{\rm{BS}}$ as shown in Fig. \ref{ProbabilityNOMA}.
In contrast, the proposed scheme still offers a considerable sum-rate gain in the larger number of antennas regime.
This is because beamwidth control provides a higher probability of forming a NOMA group and the NOMA gain can be exploited more efficiently.

\begin{figure}[t]
\centering\vspace{-2mm}
\includegraphics[width=2.7in]{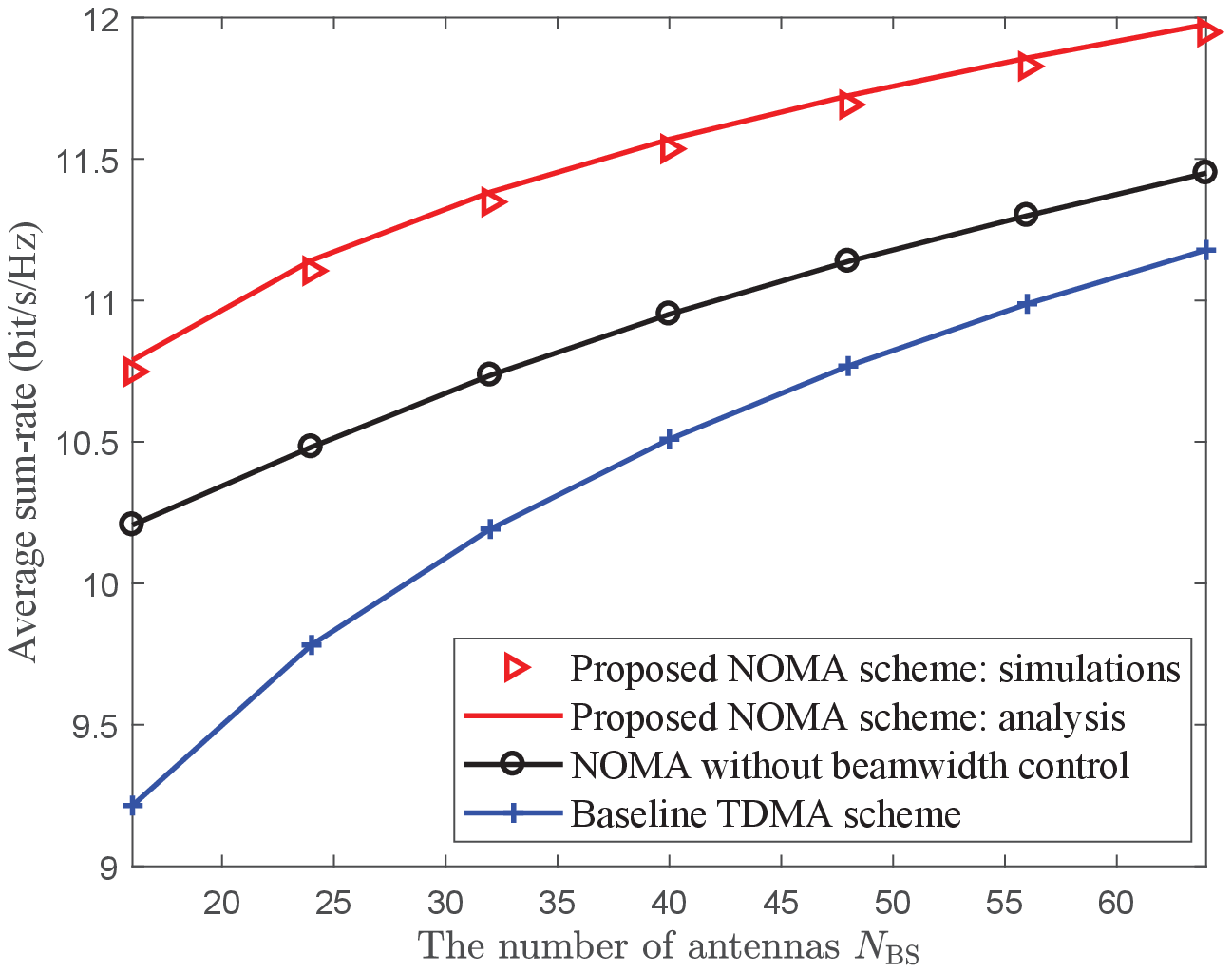}\vspace{-4mm}
\caption{Average sum-rate versus the number of antennas equipped at the BS.}\vspace{-6mm}
\label{SumRateSchemes}
\end{figure}

\section{Conclusion}

In this paper, we proposed a novel beamwidth control-based mmWave NOMA scheme and analyzed its performance.
In particular, the proposed scheme overcomes the fundamental limit of narrow beams in hybrid mmWave systems through beamwidth control, which facilitates the exploitation of NOMA transmission to enhance the system spectral efficiency.
The proposed beamwidth control method can widen the analog beamwidth via setting some of antennas at the BS as inactive.
Through characterizing the main lobe power loss due to the beamwidth control, the asymptotically optimal analog beamformer design to maximize the system sum-rate was proposed.
We analyzed the system sum-rate of the proposed NOMA scheme and derived the sufficient condition for the proposed scheme outperforming the baseline TDMA scheme.
Simulation results verified our performance analysis and demonstrated that the proposed scheme with beamwidth control offers a considerable spectral efficiency gain over the conventional TDMA and NOMA schemes without beamwidth control.


\end{document}